\documentclass[pra, singlecolumn, superscriptaddress]{revtex4}
\usepackage[latin9]{inputenc}
\usepackage{amsmath}
\usepackage{amssymb, enumerate}
\usepackage{color}
\usepackage{babel}
\usepackage{graphicx}
\usepackage{epstopdf}
\usepackage{float}
\usepackage{thmtools}
\usepackage{thm-restate}
\usepackage{hyperref}
\usepackage{cleveref}
\usepackage{enumerate}
\usepackage{lipsum}
\usepackage{fixmath}
\usepackage{dsfont}

\usepackage{amsthm}
\usepackage{amsmath}
\usepackage{amssymb}
\usepackage{tikz}
\usepackage{color}
\usetikzlibrary{matrix}
\usepackage{dsfont}
\usepackage{graphicx} 
\usepackage{stackrel}
\makeatletter

\@ifundefined{textcolor}{}
{%
 \definecolor{BLACK}{gray}{0}
 \definecolor{WHITE}{gray}{1}
 \definecolor{RED}{rgb}{1,0,0}
 \definecolor{GREEN}{rgb}{0,1,0}
 \definecolor{BLUE}{rgb}{0,0,1}
 \definecolor{CYAN}{cmyk}{1,0,0,0}
 \definecolor{MAGENTA}{cmyk}{0,1,0,0}
 \definecolor{YELLOW}{cmyk}{0,0,1,0}
}

\newenvironment{protocol*}[1]
  {
    \begin{center}
      \hrulefill\\
      \textbf{#1}
  }
  {
    \vspace{-1\baselineskip}
    \hrulefill
    \end{center}
  }
  
\theoremstyle{plain}

\theoremstyle{definition}

\theoremstyle{remark}

\newtheorem{lemma}{Lemma}  
\def\bel{\begin{lemma}}
\def\eel{\end{lemma}}
\newtheorem{theorem}{Theorem}

\newtheorem{proposition}[theorem]{Proposition}

\newtheorem{remark}[theorem]{Remark}

\newtheorem{definition}[theorem]{Definition}

\makeatother

\begin{document}

\title{Finite Device-Independent Extraction of a Block Min-Entropy Source against Quantum Adversaries}
\author{Ravishankar Ramanathan}
\email{ravi@cs.hku.hk}
\affiliation{Department of Computer Science, The University of Hong Kong, Pokfulam Road, Hong Kong}


\begin{abstract}
The extraction of randomness from weakly random seeds is a problem of central importance with multiple applications. In the device-independent setting, this problem of quantum randomness amplification has been mainly restricted to specific weak sources of Santha-Vazirani type, while extraction from the general min-entropy sources has required a large number of separated devices which is impractical. In this paper, we present a device-independent protocol for amplification of a single min-entropy source (consisting of two blocks of sufficiently high min-entropy) using a device consisting of two spatially separated components and show a proof of its security against general quantum adversaries. 
\end{abstract}



\maketitle

\section{Introduction}
The extraction of randomness from weak seeds is a topic of central importance given the utility of uniformly random bits in several cryptographic primitives as well as in randomized algorithms, physical simulations and gambling. In many applications, one requires private bits that are uniform given any side information possibly possessed by a malicious adversary. The amount of randomness in a weak seed is quantified by the min-entropy, or conditional quantum min-entropy when the adversarial side information is quantum. Specifically, an $(n, k)$ weak source is an $n$-bit string with conditional (quantum) min-entropy $k$, meaning that the probability with which an adversary can guess the source correctly is $\leq 2^{-k}$. One may also consider weak sources with more structure such as a Santha-Vazirani (SV) source \cite{SV84} wherein each bit produced by the source has some $\epsilon$ amount of randomness even conditioned on all previous bits (any SV source is also a min-entropy source however not all min-entropy sources are of SV type). The extraction of fully uniform bits from weak sources is the task termed randomness amplification or randomness extraction. 

A central well-known result is that randomness extraction is possible using classical resources only when two or more independent sources are available. In particular, with deterministic extraction using just a single min-entropy source, it is impossible to produce even one private near-uniform bit \cite{SV84}. However, independence is impossible to check or guarantee in practice, so that classical randomness extractors  make the strong assumption that independence holds in the given scenario. 

The discovery that quantum resources can help to weaken this assumption, even in a fully device-independent setting, provided a major fillip to the field \cite{CR12}. Device-Independent (DI) quantum cryptography offers the highest form of security, wherein the users (of the DI quantum random number generators) do not need to even trust the devices executing the cryptographic protocol, and can instead verify correctness and security by means of simple statistical tests on the devices. These statistical tests verify that the correlations in the devices violate a Bell inequality, and are thus genuinely quantum non-local correlations. Moreover, the adversarial side information may also be quantum in nature \cite{CSW14, KAF17, ADF+18} (or even more generally be only restricted by the no-signalling principle of relativity \cite{Barrett2005, CR12, GMTD+13, CSW16, BRGH+16, RBHH+16, HR19, RA12, RH14, GHH+14, rel-caus-2}).

Besides the practical interest in producing high-quality quantum-certified random bits, randomness amplification protocols are also of foundational interest in fundamental physics. This is because the existence of a physical process producing fully uniform bits starting from arbitrarily weak seeds may be regarded as the statement that assuming that nature obeys quantum-mechanical laws, the existence of an arbitrarily small amount of free-will implies the existence of complete freedom of choice. This was articulated as the dichotomy statement "either our world is fully deterministic or there exist in nature events that are fully random" in \cite{GMTD+13}. 

While the quantum randomness amplification of Satha-Vazirani(SV) sources using finite devices has been shown, the corresponding amplification problem of general min-entropy sources with finite devices has remained an important open question in Device-Independent Quantum Cryptography. In this paper, we present a device-independent protocol for amplification of a single min-entropy source (consisting of two blocks of sufficiently high min-entropy) using a device consisting of two spatially separated components and show a proof of its security against general quantum adversaries.

\section{Background}

Colbeck and Renner \cite{CR12} were the first to show a device-independent protocol for the amplification of randomness using quantum resources. Their protocol had the nice features of amplifying weak public sources using a device consisting of the minimum number of two separated components as well as being secure against general no-signalling side information. However, it also had a few drawbacks including being applicable only to SV sources (and only those SV sources where each bit was at most $\approx 0.058$-away from uniform), requiring a large number of measurements (growing with the security parameter), tolerating a vanishing rate of noise and producing vanishing extraction rate. Since then, advances have been made in multiple works with security proven against both quantum \cite{CSW14, KAF17} and no-signalling adversaries \cite{GMTD+13, CSW16, GHH+14, BRGH+16, RBHH+16}. The protocol in \cite{GMTD+13} allows to amplify arbitrarily weak SV sources against no-signalling side information, however it requires a large  number of spatially separated devices (polynomial in the number of bits taken from the source). In \cite{BRGH+16, RBHH+16}, practical noise-tolerant protocols were introduced for amplifying arbitrarily weak SV sources against no-signaling side information and experimentally demonstrated in \cite{our-4}. However, these latter protocols were designed to handle private rather than public sources, i.e., they assume that all of the bits produced by the source are kept private, and are never leaked to the adversary even after completion of the protocol. Recently, a protocol for handling weak public SV sources was presented in \cite{RBH21} however as in \cite{CR12} this protocol also required a large number of measurements and has vanishing noise tolerance.

Chung, Shi and Wu in \cite{CSW14} introduced the first protocol for amplifying arbitrarily weak general min-entropy sources secure against quantum adversaries and later generalised to no-signaling adversaries in \cite{CSW16}. However, their protocols had the drawback of requiring a large number of devices (for quantum adversaries, a number of devices polynomial in $1/\epsilon_f$ where $\epsilon_f$ is the final distance from uniform of the output bits, and exponential for no-signaling adversaries). Other drawbacks of their paradigm were that the security parameter was inverse polynomial in the number of bits used from the source, and the protocols had low noise-tolerance and vanishing efficiency. Nevertheless, their result was also very important from a fundamental point of view. Namely, the result shows that the presence of arbitrarily weak random process in nature together with the assumption that nature obeys the laws of quantum mechanics (or even only the no-signalling principle) imply the existence of fully random processes. 

In this paper, we present a protocol to handle min-entropy sources consisting of two blocks of sufficiently high entropy. Our setting is similar to that in \cite{KAF17} in that we will assume that the adversary is limited by quantum mechanics (and moreover holds only classical side information about the weak source), and that a Markov assumption holds - any correlations between the device and the source can be attributed to the adversarial side information. The setting may be expanded to consider a no-signalling adversary following the analysis in \cite{BRGH+16} - this allows to use only classical-proof randomness extractors to perform extraction against no-signaling adversaries, with a big penalty term on the final security parameter (and a consequent reduction in the generation rate and efficiency of the protocol). Secondly, one may consider relaxing the Markov assumption to include more general correlations between the source and device following the analysis in \cite{WBG+17}.

\section{Min-Entropy Sources}


\begin{definition} (Conditional Min-entropy).
Let $\rho_{\mathtt{X}E}\in\text{Den}(\mathtt{X}\otimes E)$, the set of normalized density operators on $\mathtt{X}\otimes E$. The min-entropy of $\mathtt{X}$ conditioned on $E$ is defined as 
\begin{equation}
    H_{\min}(\mathtt{X}|E)_{\rho}:=\max\{\lambda\in \mathbb{R} :\exists \sigma_E\in \text{Den}(E) \text { s.t. } 2^{-\lambda} \mathbb{I}_{\mathtt{X}} \otimes \sigma_E \geq \rho_{\mathtt{X}E}\}
\end{equation}
\end{definition}

\begin{definition}(Block Min-Entropy Source).
A distribution $\mathtt{X}= \mathtt{X}_1, \mathtt{X}_2,\ldots,\mathtt{X}_c$ on $\{0,1\}^{nc}$ is called a $(n,c,k_1,k_2,\ldots,k_c)$-block source if for all $i=1,\ldots,c$ we have that for all $\mathtt{x}_1\in \mathtt{X}_1,\ldots, \mathtt{x}_{i-1}\in \mathtt{X}_{i-1}$, and for any (classical) side information $\Lambda$
\begin{equation}
    H_{\min}(\mathtt{X}_i|\mathtt{X}_1=\mathtt{x}_1,\ldots,\mathtt{X}_{i-1}= \mathtt{x}_{i-1}, \Lambda)\geq k_i
\end{equation}
i.e. each block has min-entropy $k_i$ even conditioned on the previous blocks. If $k_1=k_2=\ldots=k_c = k$, we say that $\mathtt{X}$ is an $(n, c, k)$-block source. 
\end{definition}


We will also have occasion to use the smooth min-entropy which is the maximum of the min-entropy in a $\gamma$-neighborhood around the quantum state $\rho_{\mathtt{X}E}=\sum_{\mathtt{x}} p_{\mathtt{x}} |\mathtt{x}\rangle\langle \mathtt{x}|\otimes \rho_E^{\mathtt{x}}$ for $\gamma \in (0,1)$, i.e.
\begin{equation}
    H_{\min}^{\gamma}(A_{\mathtt{x}}|E)_{\rho_{\mathtt{X}E}}=\max_{\sigma_{\mathtt{X}E}\in B^{\gamma}(\rho_{\mathtt{X}E})} H_{\min} (\mathtt{X}|E)_{\sigma_{\mathtt{X}E}}
\end{equation}
where $B^{\gamma}(\rho_{\mathtt{X}E})$ is the set of sub-normalised states that are at most $\gamma$ away from $\rho_{\mathtt{X}E}$ in terms of purified distance.


$H_{\min}(\mathtt{X}|E)$ represents the number of bits that can be extracted from $\mathtt{X}$ secure against a adversary who holds $E$.
The conditional von Neumann entropy $H(\mathtt{X}|E)$ measures the number of bits that can be extracted from $\mathtt{X}$ when multiple copies of $\mathtt{X}E$ are available.

\section{Setting}
In this paper, we will be interested in a device-independent protocol using finite devices for randomness extraction from a $(n,2,k_1,k_2)$-block min-entropy source $\mathtt{X} = (\mathtt{X}_1, \mathtt{X}_2) \in \{0,1\}^{2n}$ against quantum side information. We will assume that the two blocks have sufficiently high entropy, specifically that $k_1 = O(n^{\alpha})$ for constant $0< \alpha \leq 1$ and $k_2 \geq \left(\frac{1}{2} + \delta' \right) n$ for fixed constant $0 < \delta' < 19/32$.

We consider a setting with a single min-entropy source (with two blocks) $\mathtt{X} = (\mathtt{X}_1, \mathtt{X}_2) \in \{0,1\}^{2n}$ and an untrusted device with at least two separated components. Both the source and the device can be correlated with classical side information $\Lambda$ held by an adversary. The adversary furthermore holds a purification $E$ of the quantum state of the device. During the protocol, the source produces inputs $\mathcal{X}^N \mathcal{Y}^N$ for the device which produces outputs $\mathcal{A}^N \mathcal{B}^N$. Here we have denoted $\mathcal{X}^N = X_1, \ldots, X_N$ where $X_j$ denotes the input for the $j$-th run of the protocol and analogously for $\mathcal{Y}^N$, and similarly we have denoted $\mathcal{A}^N = A_1, \ldots A_N$ where $A_j$ denotes the output of Alice's device in the $j$-th run of the protocol and analogously for $\mathcal{B}^N$. An independent-source extractor $\text{Ext}$ produces the final output bits of the protocol $R$ using as its inputs the outputs of the device $\mathcal{A}^N\mathcal{B}^N$ and the second block of bits $\mathtt{X}_2$ from the min-entropy source, i.e., $R = \text{Ext}\left(\mathcal{A}^N\mathcal{B}^N, \mathtt{X}_2 \right)$.

\subsection{Assumptions}
We make the following assumptions in the protocol.
\begin{enumerate}
\item The adversary is limited by quantum mechanics, specifically the adversary holds a purification $E$ of the initial quantum state of the device held by the honest players.

\item The adversary holds classical side information $\Lambda$ about the min-entropy source.

\item The untrusted device has two separated components that are non-signalling (shielded) with respect to each other. 

\item The source is a $(n, 2, k_1, k_2)$-block source with $k_1 = O(n^{\alpha})$ for constant $0< \alpha \leq 1$ and $k_2 \geq \left(\frac{1}{2} + \delta' \right) n$ for fixed constant $0 < \delta' < 19/32$.

\item 
\label{assum:Markov} While the device produces outputs, it holds that
\begin{equation}
I\left(\mathcal{A}^{l-1}\mathcal{B}^{l-1}:X_lY_l | \mathcal{X}^{l-1} \mathcal{Y}^{l-1}E \Lambda \right) = 0,
\end{equation}
and after the device stops producing outputs, it holds that
\begin{equation}
I\left(\mathtt{X}_2 : \mathcal{A}^N\mathcal{B}^N|\mathcal{X}^N\mathcal{Y}^NE\Lambda \right) =0.
\end{equation}
Here $\mathcal{A}^{l-1} = A_1, A_2, \ldots, A_{l-1}$ and similarly for the other random variables. 

\item If the devices running the protocol are later reused, they do not leak information about previously run protocols \cite{BCK13}. 

\end{enumerate}

The Assumption \ref{assum:Markov} is known as the Markov assumption \cite{FPS16, KAF17} and it implies that the Markov model for quantum-proof multi-source extractors can be used \cite{FPS16}. In this setting, the source and the device are considered to be independent conditioned on the adversarial side information which may be quantum. This assumption states that the source block $\mathtt{X}_2$ and the outputs of the device $\mathcal{A}^N\mathcal{B}^N$ with the side information (and previous bits from the min-entropy source) $\mathcal{X}^N\mathcal{Y}^NE\Lambda$ form a Markov chain. Note that here $\mathtt{X}_2$ and $\mathcal{A}^N\mathcal{B}^N$ are classical states while $E$ represents a quantum state. We say that a $ccq$-state $\rho_{C_1C_2Q}$ belongs to the Markov model if $C_1 \leftrightarrow Q \leftrightarrow C_2$ forms a Markov chain, i.e., the conditional mutual information $I(C_1:C_2 | Q)_{\rho} = 0$. In other words, we assume that any correlation between the bits taken from the source and the bits produced by the device $D=(D_1,D_2)$, are due to the previous bits taken from the source and Eve's side information. A result proven in \cite{FPS16} then shows that provided that $\mathcal{A}^N\mathcal{B}^N$ and $\mathtt{X}_2$ have sufficient amount of entropy conditioned on these variables, any (strong) multi-source extractor remains a (strong) quantum-proof multi-source extractor in the Markov model, so that the randomness can be extracted (with some loss in parameters). The recent development of such extractors is crucial in allowing us to extract randomness from a single weak source. The Bell test on the device ensures that the output string $\mathcal{A}^N\mathcal{B}^N$ is independent of the block $\mathtt{X}_2$ conditioned on the inputs and any adversarial side information. The assumption \ref{assum:Markov} also states that the inputs $X_l, Y_l$ in the $l$-th round of the protocol for $l \in [N]$ do not reveal any new information about the previous outcomes $\mathcal{A}^{l-1}\mathcal{B}^{l-1}$ (where $\mathcal{A}^{l-1} = A_1,A_2,\ldots, A_{l-1}$) other than what was already available through the previous inputs $\mathcal{X}^{l-1}\mathcal{Y}^{l-1}$ (where $\mathcal{X}^{l-1} = X_1,X_2,\ldots, X_{l-1}$) and through the side information $E, \Lambda$. This assumption is crucial to apply the entropy accumulation theorem \cite{DOR16, ADF+18, KAF17} and derive a lower bound on the (conditional smooth min) entropy of $\mathcal{A}^N\mathcal{B}^N$ conditioned on the input bits $\mathcal{X}^N\mathcal{Y}^N$ and adversarial variables $E\Lambda$.

\section{The Protocol}


\textbf{Parameters}:\\
$(n, 2, k_1, k_2)$: Two-block min-entropy source with $k_1 = O(n^{\alpha})$ for constant $0< \alpha \leq 1$ and $k_2 \geq \left(\frac{1}{2} + \delta' \right) n$ for fixed constant $0 < \delta' < 19/32$. 
$D=(D_1,D_2):$ untrusted device of $2$ separated components\\
$N:$ number of the rounds\\
$G_{\text{exp}}:$ Lower bound on the expected score in a game $G$ for an honest implementation\\
$\delta:$ Confidence interval for the estimation \\
$\text{Ext}_{\text{T}}:\{0,1\}^{n}\times\{0,1\}^{d}\rightarrow\{0,1\}^{m}:$ Trevisan's quantum-proof $(k, \epsilon)$-strong seeded extractor with $k = n^{\gamma} m + 8 \log(m/\epsilon) + O(1)$, $d = O(\log n)$ and $\epsilon = n^{-\Omega(1)}$ (Corollary 5.6 of \cite{DPVR12}).
\\
$\text{Ext}_{\text{R}}: \{0,1\}^N \times \{0,1\}^n \rightarrow \{0,1\}^M$: Raz's $\left(k_1, k_2, \epsilon_{\text{Ext}} \right)$ quantum-proof randomness extractor strong in both inputs separately \cite{FPS16}. \\ 
\begin{figure}[h]
	\begin{protocol*}{Protocol I}
		\begin{enumerate}
	\item[1] The honest players use the $n$ bits $\mathtt{X}_1$ along with bit strings $j$ of length $d$ for $j \in [2^d]$ as inputs to the Trevisan extractor $\text{Ext}_{\text{T}}$ to obtain $2^d$ bit strings $S_j = \text{Ext}_{\text{T}}(\mathtt{X}_1, j)$ each of length $m$. They choose an arbitrary substring $S'_j$ of length $m' = O(\log n)$ for each $j \in [2^d]$.  The $S'_j$ constitute a quantum-somewhere-random source with full support on all $2^{m'}$ strings, where each bit is $\epsilon$-away from uniform with respect to device. 	
	\item[2] The players use the obtained $2^d \cdot m'$ bits as inputs $\left(X_{j,k_j}, Y_{j,k_j} \right)$ in $N := 2^d \cdot (m'/2)$ sequential runs of an MDL-Hardy game, where $j \in [2^d]$ indicates the blocks and $k_j \in [m'/2]$ indicates the run within the $j$-th block. That is, the players use the bits $S'_1$ to choose inputs for the first $m'/2$ runs, $S'_2$ for the next (sequential) block of $m'/2$ runs, and so on till the $2^d$-th block of $m'/2$ runs. They record the corresponding outputs $\left(A_{j,k_j}, B_{j, k_j} \right)$ in each run. We denote the output bit string $\mathcal{K}^N = \big\{K_{1,1}, \ldots, K_{2^d, m/2} \big\}$ for $K = A, B$, and the inputs as $\mathcal{L}^N = \big\{L_{1,1}, \ldots, L_{2^d, m/2} \big\}$ for $L  = X, Y$. 
	\item[3] The players test for the violation of the corresponding MDL-Hardy parameter $M_{\epsilon}$ in each of the $2^d$ blocks of $m'/2$ runs each. That is, the players test that the average value $L_{m'}^{j} = \frac{2}{m'}\sum_{k_j=1}^{m'/2} M_{\epsilon}^{j,k_j}$ of the MDL-Hardy parameter in the $k_j$-th run within the $j$-th block for $j \in [2^d], k_j \in [m'/2]$ satisfies $L_{m'}^{j} \geq \delta$ for each $j \in [2^d]$ where $\delta > 0$ is a fixed constant. They abort the protocol if $L_{m'}^{j} < \delta$ for any $j \in [2^d]$. 
	\item[4] The honest players then use the $n$ bits $\mathtt{X}_2$ from the two-block min-entropy source along with Alice and Bob's output bit strings $\mathcal{A}^N\mathcal{B}^N$ as inputs to the Raz extractor $\text{Ext}_{\text{R}}$ to produce the final output bits as 
	\begin{equation}
	R := \text{Ext}_{\text{R}}\left(\mathcal{A}^N\mathcal{B}^N, \mathtt{X}_2 \right).
	\end{equation} 
	 
    
    
    
    
\end{enumerate}
	\end{protocol*}
	\caption{\label{protocol1sDI} Protocol I for Device-Independent Randomness Amplification of a Min-Entropy Source.}
\end{figure}

\section{Quantum-Somewhere-Random-Source}
The first step of the protocol is a classical pre-processing procedure where the honest parties apply a seeded extractor (specifically we will use Trevisan's extractor which was proven to be quantum-proof \cite{DPVR12}) $\text{Ext}_{\text{T}}$ on $\mathtt{X}_1$ (an $n$-bit string) together with seeds $j$ that are $d$-bit strings. Since we do not have an independent random seed, we will enumerate all $2^d$ possible bit strings $j$ and create a set of $2^d$ random variables $S_j = \text{Ext}_{\text{T}}(\mathtt{X}_1, j)$ for $j = 1, \ldots, 2^d$. This procedure creates a somewhere-random-source, i.e., a source in which (at least) one of the $S_j$ is (close-to-)random (this procedure was also used in the protocol in \cite{CSW14}) .

\begin{definition}(Quantum-Somewhere-Random Source).
A classical-quantum-state $\rho\in \text{Den}(H_{S_1}\otimes H_{S_2}\otimes\ldots\otimes H_{S_{2^d}}\otimes H_E)$ with classical $S_1,S_2,\ldots,S_{2^d}\in\{0,1\}^m$ and quantum $E$ is a $(2^d,m)$-quantum-somewhere-random source against $E$ if there exists a $j\in [2^d]$ such that 
\begin{equation}
    \rho_{S_jE}=\rho_{Um}\otimes\rho_E,
\end{equation}
where $\rho_{Um}$ is the fully mixed state on a system of dimension $2^m$.
We say that $\rho$ is a $(2^d,m,\epsilon)$-quantum somewhere-random source if there exists a $j\in[2^d]$ such that 
\begin{equation}
    \frac{1}{2}\big\|\rho_{S_jE}-\rho_{Um}\otimes\rho_E \big\| \leq \epsilon.
\end{equation}
\end{definition}
\begin{proposition}(\cite{CSW14}).
Let $\text{Ext:}\{0,1\}^n\times\{0,1\}^d\rightarrow\{0,1\}^m$ be a quantum-proof ${(k,\epsilon)}$-strong extractor. Let $\rho_{\mathtt{X}_1E}$ be a cq-state with $H_{\min}(\mathtt{X}_1|E)\geq k$.  For every $j\in\{0,1\}^d$, let $S_j=\text{Ext}(\mathtt{X}_1,j)$. Then the cq-state
\begin{equation}
    \rho_{S_1S_2\ldots S_{2^d}E}:=\sum_{\mathtt{x}} p_{\mathtt{x}}|S_1\rangle\langle S_1|\otimes\ldots\otimes|S_{2^d}\rangle\langle S_{2^d}|\otimes\rho_{\mathtt{x}}^E
\end{equation}
is a $(2^d,m,\epsilon)$-quantum-somewhere-random source. Furthermore, the expectation of $\frac{1}{2}\big\|\rho_{S_jE}-\rho_{Um}\otimes\rho_E\big\|$ over a uniform random index $j\in\{0,1\}^d$ is at most $\epsilon$.
\end{proposition}
The proof of this proposition is rather straightforward and comes simply from the fact that the output of the extractor being $\epsilon$-away from uniform for a uniformly random seed implies that
\begin{equation}
\sum_{j=1}^{2^d} \frac{1}{2^d} \big\|\rho_{\text{Ext}(\mathtt{X}_1,j)E}-\rho_{Um}\otimes\rho_E\big\|\leq \epsilon.
\end{equation}
So that at least for one the $j^* \in \{0,1\}^d$, it holds that $\big\|\rho_{\text{Ext}(\mathtt{X}_1,j^*)E}-\rho_{Um}\otimes\rho_E\big\| \leq \epsilon$. 

\begin{remark}
In fact, one can go further and define $S := \big\{ j \; | \big\|\rho_{\text{Ext}(\mathtt{X}_1,j)E}-\rho_{Um}\otimes\rho_E\big\| \leq \sqrt{\epsilon} \big\}$. This gives that
\begin{eqnarray}
&&\sum_{j \in S} \frac{1}{2^d} \big\|\rho_{\text{Ext}(\mathtt{X}_1,j)E}-\rho_{Um}\otimes\rho_E \big\| + \sum_{j \notin S} \frac{1}{2^d} \big\|\rho_{\text{Ext}(\mathtt{X}_1,j)E}-\rho_{Um}\otimes\rho_E \big\|  \leq  \epsilon, \nonumber \\
&&0 + \frac{1}{2^d} (2^d - |S|) \sqrt{\epsilon} \leq \epsilon \implies \quad
|S| \geq  2^d(1 - \sqrt{\epsilon}).
\end{eqnarray}
In other words, for at least $2^d(1 - \sqrt{\epsilon})$ indices $j \in \{0,1\}^d$, it holds that $\big\|\rho_{\text{Ext}(\mathtt{X}_1,j)E}-\rho_{Um}\otimes\rho_E \big\| \leq \sqrt{\epsilon}$. While this observation in principle means that a linear fraction of the $S_j$ are random, it must be noted that the random variables $S_j$ are correlated with each other, so that care must be taken while using these as inputs in sequential runs of the protocol (as we shall see in the following section). 
\end{remark}

In the first step of the protocol, we use the extractor from Thm. [\ref{thm-first-step-2}] which works if the source has sufficiently high min-entropy $H_{\min}(\mathtt{X}_1|\Lambda) = n^{\alpha}$ for constant $0 < \alpha \leq 1$. In this case, the seed length is $d = O(\log n)$ and we obtain the following statement. 

\begin{proposition}
Let $\mathtt{X}=(\mathtt{X}_1, \mathtt{X}_2)$ on $\{0,1\}^{2n}$ be a $(n, 2, k_1, k_2)$-block min-entropy source with $k_1 = O(n^{\alpha})$ for constant $0< \alpha \leq 1$ i.e.,
\begin{eqnarray}
	H_{\min}(\mathtt{X}_1|\Lambda) &\geq & k_1, \nonumber \\
    H_{\min}(\mathtt{X}_2|\mathtt{X}_1=\mathtt{x}_1, \Lambda) &\geq & k_2  \quad \forall x_1.
\end{eqnarray}
for any classical side information $\Lambda$ held by the adversary. Let $\text{Ext}_{\text{T}}$ be the extractor from Thm. \ref{thm-first-step-2}. 
For every $j\in\{0,1\}^d$ with $d = O(\log n)$ let $S_j=\text{Ext}_{\text{T}}(\mathtt{X}_1,j)$. Then the classical-quantum-state
\begin{equation}
    \rho_{S_{1}S_{2}\ldots S_{2^d}E}=\sum_x p_{x}|S_{1}\rangle\langle S_{1}|\otimes|S_{2}\rangle\langle S_{2}|\otimes\ldots\otimes |S_{2^d}\rangle\langle S_{2^d}|\otimes \rho_{x}^E
\end{equation}
is a $(2^d,m,\epsilon)$-quantum-somewhere-random source, i.e. $\exists j^* \in[2^d]$ such that $\big\|\rho_{S_{j^*}E}-\rho_{Um}\otimes\rho_E \big\| \leq \epsilon$, where $\epsilon = n^{-c_1}$ for constant $c_1 > 0$ and $m = n^{\alpha - \gamma} - o(1)$ for constant $0 < \gamma < \alpha$. 
\end{proposition}

\section{Using the Somewhere-Random-Source in a Bell test}
\label{sec:SRS-Bell}
At this point, we have a $(2^d,m,\epsilon)$-quantum-somewhere-random source with $\epsilon = n^{-c_1}$ for constant $c_1 > 0$. By the property of the quantum-somewhere-random source, we have that there exists $j^* \in[2^d]$ such that $\big\|\rho_{\text{Ext}(X,j^*)E}-\rho_{Um}\otimes\rho_E\big\| \leq \epsilon$ (there exists index $j^*$ such that the $m$-bit string $S_{j^*}$ is $\epsilon$-away from uniform). By the fact that every $\epsilon$-extractor is also an $\epsilon$-disperser, this means that for the index $j^*$, the extractor has nearly full support on outputs of length $m$, specifically the support of the output of the extractor obeys $|\text{Supp}\left(\text{Ext}_{\text{T}}(X,j^*) \right)| \geq (1- \epsilon) \cdot 2^m$. With $\epsilon = n^{-c_1}$ and by the construction of the Trevisan extractor, we choose an output length $m' = O(\log n)$ such that the output of the extractor has full support on this output, i.e., outputs every possible $m'$-bit string with positive probability. Specifically, we have the requirement that $2^{m-m'} > \epsilon \cdot 2^m$ so that $m' < -\log_2 \epsilon = O(\log n)$ for $\epsilon = n^{-c_1}$. Note that we can choose this such that $m' = c_2 \cdot d$ for any constant $c_2 > 1$. That is, we obtain a $(2^d, m', \epsilon)$-quantum-somewhere-random source with $\epsilon = n^{-c_1}$ where each of the $m'$-bit strings appears with positive probability for an index $j^* \in [2^d]$ (the extractor acts as a disperser for this value of $m'$ and outputs every possible $m'$-bit string with positive probability).

The second step in the protocol is to use this source as the input for $\big\lfloor 2^d \cdot m'/p \big \rfloor$ rounds of a suitably designed Bell test with $p$-bit inputs per round. We will illustrate Protocol I with the measurement-dependent-locality (MDL) version of the two-player Hardy paradox where $p=2$, and without loss of generality we will take $m'$ to be even. 

The measurement-dependent-locality (MDL) version of a Bell test considers a modified Bell inequality that tolerates inputs that are biased away from uniform and correlated with the device running the Bell test, i.e., in the situation when perfect measurement-independence is not available. An MDL inequality considers the joint distributions $P_{A,B,X,Y}(a,b,x,y)$ rather than the conditional behaviors $P_{A,B|X,Y}(a,b|x,y)$ considered in a traditional Bell inequality, and is designed to detect non-locality in any situation in which the inputs are not fully deterministic i.e., when $l \leq P_{X,Y}(x,y) \leq h$ holds for $0 < l, h < 1$. For more details on measurement-dependent-locality inequalities, we refer the reader to \cite{PRB+14, KAF17, our-4}. 

In each round of the two-player MDL-Hardy test, each player receive a single bit input ($x, y \in\{0,1\}$) and produces a single bit output ($a, b\in\{0,1\}$). The goal of the players is to maximize the value of the MDL-Hardy parameter
\begin{equation}
M_{\epsilon} := (1/2-\epsilon)^2 P_{A,B,X,Y}(0,0,0,0) - (1/2+\epsilon)^2 \left[P_{A,B,X,Y}(0,1,0,1) + P_{A,B,X,Y}(1,0,1,0) + P_{A,B,X,Y}(0,0,1,1) \right].
\end{equation}
One can readily check that $M_{\epsilon} \leq 0$ holds for classical theories (Local Hidden Variable models) while a quantum strategy exists that can achieve $P^{Q}_{A,B,X,Y}(0,0,0,0) > 0$ with $P^{Q}_{A,B,X,Y}(0,1,0,1) = P^{Q}_{A,B,X,Y}(1,0,1,0)=P^{Q}_{A,B,X,Y}(0,0,1,1)= 0$. We will use the $2^d \cdot m'$ bits from the somewhere-random-source as the inputs for $N := 2^d \cdot (m'/2)$ sequential runs of the MDL-Hardy test. That is we use the bits to choose the inputs $X_{j,k_{j}}, Y_{j, k_{j}}$ for each of $m'/2$ runs (i.e., $k_{j} \in [m'/2]$) in the $j$-th block for $j = 1, \dots, 2^d$. 

In \cite{KAF17}, the following bound was derived on the randomness certified per round in the MDL-Hardy test when $M_{\epsilon} > 0$. 
\begin{lemma}[\cite{KAF17}]
Suppose that in  the $k_j$-th round in the $j$-th block of Protocol I, the true value of the MDL-Hardy parameter is $M_{\epsilon} > 0$. Then the bound
\begin{equation}
H(A_{j,k_j} B_{j,k_j}|X_{j,k_j}Y_{j,k_j}E\Lambda) \geq 1 - h\left(\frac{1}{2} + \frac{1}{(1/4-\epsilon^2)^2}\sqrt{M_{\epsilon} \left(M_{\epsilon} +(1/4-\epsilon^2)^2 \right)} \right)
\end{equation}
on the conditional von Neumann entropy of the outputs $A_{j,k_j}, B_{j,k_j}$ holds, where $h(p) = - p \log p - (1-p) \log (1-p)$ denotes the binary entropy.
\end{lemma}

It was also shown, using the Entropy Accumulation Theorem (EAT) \cite{ADF+18, DOR16}, that the total entropy generated in $m'/2$ runs where the input is $\epsilon$-away from uniform is linear in $m'$. We use this fact for the $m'/2$ runs within the $j^*$ block while in the runs within the blocks $j \neq j^*$ we simply bound the entropy by $0$, i.e., $H(A_{j,k_j} B_{j,k_j}|X_{j,k_j}Y_{j,k_j}E) \geq 0$. This is because in these rounds, the inputs may not be random so that the Bell violation could in principle be simulated by local deterministic boxes that do not contain any randomness. 

%
%
%

Specifically, let us denote $\perp$ as the event of aborting the protocol and $\overline{\perp}$ as the complementary event of not aborting the protocol,
\begin{equation}
\overline{\perp} := \bigg\{ \wedge_{j=1}^{2^d} \left( L_{m'}^j \geq \delta \right) \bigg\}.
\end{equation}
Let $\rho = \rho^{\mathcal{A}^N\mathcal{B}^N\mathcal{X}^N\mathcal{Y}^N\mathcal{D}^NE\Lambda}$ denote the joint state of the devices held by the honest parties and Eve at the end of the $N := 2^d \cdot (m'/2)$ runs of the protocol. Here $\mathcal{D}^N$ denotes the classical random variable that stores the results of the (Bell) winning condition in the $N$ rounds. Let $\rho_{|\overline{\perp}}$ denote the state conditioned on the event of not aborting the protocol. The Entropy Accumulation Theorem gives a way of lower bounding the $\gamma$ conditional smooth min-entropy $H_{\min}^{\gamma}\left(\mathcal{A}^N\mathcal{B}^N|\mathcal{X}^N\mathcal{Y}^NE \Lambda \right)_{\rho_{|\overline{\perp}}}$ of this state for any $\gamma \in (0,1)$. In particular, building on the statement shown in \cite{KAF17} we have the following (the details are deferred to section \ref{sec:Entropy-Acc}).
\begin{proposition}
\label{prop:rand-acc}
Let $\rho = \rho^{\mathcal{A}^N\mathcal{B}^N\mathcal{X}^N\mathcal{Y}^N\mathcal{D}^NE\Lambda}$ denote the joint state of the devices held by the honest parties and Eve at the end of the $N$ runs of the protocol, $\overline{\perp}$ denote the event of not aborting and let $\rho_{|\overline{\perp}}$ denote the state conditioned on the event of not aborting the protocol. Then for any $\epsilon_{EA}, \gamma \in (0,1)$, either the protocol aborts with probability greater than $1 - \epsilon_{EA}$ or it holds that
\begin{equation}
H_{\min}^{\gamma}\left(\mathcal{A}^N\mathcal{B}^N|\mathcal{X}^N\mathcal{Y}^NE\Lambda \right)_{\rho_{|\overline{\perp}}} = \Omega(m').
\end{equation}

\end{proposition}

\section{Extraction}
From the entropy accumulation statement, we know that if the protocol aborts with probability $< 1 - \epsilon_{EA}$ then $\mathcal{A}^N\mathcal{B}^N$ constitutes a min-entropy source conditional on any quantum side information $E$ held by an adversary, any classical side information $\Lambda$ held by the adversary about the min-entropy source, as well as the inputs $\mathcal{X}^N\mathcal{Y}^N$ used during the protocol. That is
\begin{equation}
H_{\min}^{\gamma}(\mathcal{A}^N\mathcal{B}^N|\mathcal{X}^N\mathcal{Y}^NE\Lambda)\geq k_1,
\end{equation}
for $k_1 = O(m') = O(\log N)$ and $\gamma \in (0,1)$. This comes from the fact that $N = 2^d \cdot (m'/2)$ with $d = O(\log n)$ and $m' = O(\log n)$, where $m' = c_2 \cdot d$ for constant $c_2 > 1$. 

In the last step of the protocol, the parties use the second block of $n$ bits $\mathtt{X}_2$ from the min-entropy source that has sufficiently high min-entropy $k_2$ conditioned on $\mathcal{X}^N\mathcal{Y}^N$, $E$ and $\Lambda$. Specifically
\begin{equation}
H_{\min}(\mathtt{X}_2|\mathcal{X}^N\mathcal{Y}^NE\Lambda)\geq k_2,
\end{equation}
where $k_2 \geq \left(\frac{1}{2} + \delta' \right) n$
for suitable constant $0 < \delta' < 19/32$. 

Since $k_1 = O(\log n)$ we will need $k_2$ to be sufficiently high in order to use an independent-source extractor. Specifically we will use Raz's independent-source extractor $\text{Ext}_{\text{R}}$ from Thm. (\ref{thm-second-ext}). In order to use this, as stated earlier, we make the Markov assumption \ref{assum:Markov} as in \cite{KAF17}, i.e., we assume that 
\begin{equation}
    I\left(\mathtt{X}_2:\mathcal{A}^N\mathcal{B}^N|\mathcal{X}^N\mathcal{Y}^NE\Lambda \right)=0.
\end{equation}
Recall that the Markov model is the assumption that the two sources $\mathtt{X}_2$ and $\mathcal{A}^N\mathcal{B}^N$ with the side information (and previous bits from the min-entropy source) $\mathcal{X}^N\mathcal{Y}^NE\Lambda$ form a Markov chain. 

We now use the following Lemmas proven in \cite{FPS16} in conjunction with the Raz extractor from Thm. \ref{thm-second-ext}.
\begin{lemma}[\cite{FPS16}]
Any $(k_1, k_2, \epsilon)$-strong two-source extractor is a $\left(k_1 + \log \frac{1}{\epsilon}, k_2 + \log \frac{1}{\epsilon}, \sqrt{3 \epsilon 2^{m-2}}\right)$-strong quantum-proof two-source extractor in the Markov model, where $m$ is the extractor output length. 
\end{lemma}
\begin{lemma}[\cite{FPS16}]
\label{lem:fin-out-dist}
Let $\text{Ext}:\{0,1\}^n \times \{0,1\}^d \rightarrow \{0,1\}^m$ be a $(k_1, k_2, \epsilon)$ quantum-proof two-source extractor in the Markov model, strong in the source $X_i$. For any Markov state $\rho_{X_1X_2C}$ with $H_{\min}^{\epsilon_s}(X_1|C)_{\rho} \geq k_1 + \log \frac{1}{\epsilon} + 1$ and $H_{\min}(X_2|C)_{\rho} \geq k_2 + \log \frac{1}{\epsilon} + 1$, it holds that
\begin{equation}
\frac{1}{2} \big\| \rho_{\text{Ext}(X_1,X_2)X_iC} - \rho_{U_m} \otimes \rho_{X_iC} \big\| \leq 6(\epsilon_s + \epsilon).
\end{equation}
\end{lemma}
We thus obtain the following statement.

\begin{proposition}
\label{prop:soundness}
Let $\text{Ext}_{\text{R}}: \{0,1\}^N \times \{0,1\}^{n} \rightarrow \{0,1\}^M$ be Raz's $(k_1, k_2, \epsilon_{\text{Ext}})$ quantum-proof two-source extractor in the Markov model (Thm. \ref{thm-second-ext}), strong in the second input, such that
\begin{eqnarray}
k_1 & \geq & (m'/2) \cdot g\left(\gamma, \epsilon_{EA}, \delta, m'/2, \epsilon \right) - \log(1/\epsilon_{\text{Ext}}) - 1 \nonumber \\
k_2 &\geq & \left(\frac{1}{2} + \delta'' \right)n + 3 \log n + \log N - \log(1/\epsilon_{\text{Ext}}) - 1,
\end{eqnarray}
with $n \geq 6 \log n + 2 \log N$ and $0 < \delta'' < 19/32$. 
Consider the Protocol I using $\text{Ext}_{\text{R}}$ and $\epsilon_{EA}, \gamma \in (0,1)$. Then either the Protocol I aborts with probability $\geq 1 - \epsilon_{EA}$ or for the output $R$ of the extractor together with the side information  $\Sigma = \mathcal{X}^N\mathcal{Y}^NE\Lambda$ it holds that 
\begin{equation}
\frac{1}{2} \big\| \rho_{R\Sigma} - \rho_{U_M} \otimes \rho_{\Sigma} \big\| \leq 6(\gamma + \epsilon_{\text{Ext}}).
\end{equation}
\end{proposition}
\begin{proof}
From Prop. \ref{prop:rand-acc}, we know that either the protocol aborts with probability $\geq 1 - \epsilon_{EA}$ or the smooth entropy of the outputs is lower bounded by $(m'/2) \cdot g\left(\gamma, \epsilon_{EA}, \delta, m'/2, \epsilon \right)$. Since $N = 2^d \cdot (m'/2)$ with $d, m = O(\log n)$ we have that $N = O(n \log n)$ so that the requirement $n \geq 6 \log n + 2 \log N$ is met. Also, we have $m', d = O(\log n)$ with $m' = c_2 \cdot d$ for a constant $c_2 > 1$. So that the requirement in Thm. \ref{thm-second-ext} of $k_1 \geq \frac{163}{32} \log \left( \left(1 + \frac{3 \delta'}{19} \right) n - k_2 \right)$ is met for a suitable constant $c_2 > 1$.

From the requirement on the min-entropy source we have that $H_{\min}(\mathtt{X}_2|\mathcal{X}^N\mathcal{Y}^NE\Lambda)\geq  \left(\frac{1}{2} + \delta' \right) n  \geq \left(\frac{1}{2} + \delta'' \right) n + 3 \log n + \log N$ for suitable constants $0<\delta'' < \delta' < 19/32$. By assumption the state of the source, devices and Eve obeys the Markov property so that the independent-source extractor from Thm. \ref{thm-second-ext} can be applied. From Lemma \ref{lem:fin-out-dist} \cite{FPS16} we have that the $M$ output bits $R$ from the extractor are $6(\gamma + \epsilon_{\text{Ext}})$ away from uniform. 

\end{proof}

%

\section{Security}
We are now ready to state the soundness statement regarding the distance from uniform of the final output bits of the Protocol I. We use the standard definition of composable security where conditioned upon the protocol not aborting, the final output bits $R$ are $\epsilon_f$-close to uniform. Formally, we have
\begin{definition}
The amplification protocol that outputs a bit string $R$ of length $M$ is said to be $\epsilon_f$-secret against an adversary holding side information $\mathcal{E}$ if
\begin{equation}
\label{eq:soundness}
\left(1 - \text{Pr}\left[\perp \right]\right) \big\| \rho_{R\mathcal{E}} - \rho_{U^{M}} \otimes \rho_{\mathcal{E}} \big\| \leq \epsilon_f,
\end{equation}
where $\rho_{U_M}$ denotes the fully mixed state of dimension $2^M$, and $\text{Pr}\left[\perp\right]$ denotes the probability that the protocol aborts.
\end{definition}

\begin{theorem}
\label{thm:soundness}
For any $\epsilon_{EA}, \gamma \in (0,1)$ the Protocol I is $\epsilon_f$ secret with $\epsilon_f = 12 \left(\epsilon_{\text{Ext}} + \gamma \right) + \epsilon_{EA}$.
\end{theorem}
\begin{proof}
We see that in case the protocol aborts with probability $\geq 1 - \epsilon_{EA}$, that $\left(1 - \text{Pr}\left[\perp \right]\right) \leq \epsilon_{EA}$ so that Eq.(\ref{eq:soundness}) holds for $\epsilon_f = 12 \left(\epsilon_{\text{Ext}} + \gamma \right) + \epsilon_{EA}$. When the protocol aborts with probability $< 1 - \epsilon_{EA}$ we have that $ \big\| \rho_{R\mathcal{E}} - \rho_{U^{M}} \otimes \rho_{\mathcal{E}} \big\|  \leq 12 \left(\epsilon_{\text{Ext}} + \gamma \right)$ from Prop. \ref{prop:soundness} so that again Eq.(\ref{eq:soundness}) holds for $\epsilon_f = 12 \left(\epsilon_{\text{Ext}} + \gamma \right) + \epsilon_{EA}$.

\end{proof}

We turn to the notion of Completeness, i.e., that there exists an honest implementation in which the protocol aborts with negligible probability. If the device $D = (D_1, D_2)$ implements $N$ independent, ideal measurements on the ideal state $|\psi_H^{\otimes N} \rangle$ for the Hardy paradox, the expected average value of the MDL-Hardy parameter $L_{m'}^j = \frac{2}{m'}\sum_{k_j = 1}^{m'/2} M_{\epsilon}^{j,k_j}$ in each of the $j \in [2^d]$ blocks is at least $G_{\text{exp}}$. We can calculate the probability that $L_{m'}^j < \delta$ for constant $\delta > 0$ and any $j$ using the Azuma-Hoeffding inequality as
\begin{equation}
\text{Pr}\left[L_{m'}^j < \delta \right] = \text{Pr}\left[\left(G_{\text{exp}}-L_{m'}^{j} \right) > \left(G_{\text{exp}} - \delta \right) \right] \leq \exp\left[-m' (G_{\text{exp}} - \delta)^2/3 \right]. 
\end{equation}
The probability that the protocol aborts for an honest implementation can then be calculated as
\begin{eqnarray}
\text{Pr}\left[\perp \right] = \text{Pr}\left[ \vee_{j=1}^{2^d} \left(L_{m'}^j < \delta \right) \right] \leq 2^d   \exp\left[-m' (G_{\text{exp}} - \delta)^2/3 \right],
\end{eqnarray}
where we have used the union bound. Since we have $m' = c_2 \cdot d$ for $c_2 > 1$ we can choose sufficiently large $c_2$ such that $\text{Pr}\left[\perp \right] = \exp\big\{-O(\log n)\big\} \rightarrow 0$ so that the probability that the protocol aborts for an honest implementation is negligible.

\section{Discussion}
We have presented a protocol for quantum randomness amplification of a two-block min-entropy source where each of the blocks has sufficiently high min-entropy. A few points regarding the structure of the protocol and the requirements on the input min-entropy must be considered for future improvements. In \cite{BCK13}, a crucial weakness of device-independent protocols that rely on public communication between separated devices was pointed out. Namely, that untrusted devices may record their inputs and outputs in their memory and reveal information about previous outputs via publicly disclosed outputs during later runs. In view of this it is natural to consider why the strategy pursued here, wherein input randomness is present in one of $2^d$ blocks, can work. This is where the Markov chain assumption plays a role, namely the inputs $X_l, Y_l$ in the $l$-th round of the protocol for $l \in [N]$ do not reveal any new information about the previous outcomes $\mathcal{A}^{l-1}\mathcal{B}^{l-1}$ other than what was already available through the previous inputs $\mathcal{X}^{l-1}\mathcal{Y}^{l-1}$ and through the side information $E, \Lambda$. In other words, the weak source does not update its state depending on previous outcomes from the device. 

We also see that it is not sufficient to test for the overall Bell violation over all the $N$ runs and one has to rather test for the Bell violation in each of the $2^d$ blocks. This is because an overall Bell violation over $N$ runs may be simulated by a classical device, since $2^d - 1$ blocks may not contain random inputs. A typical testing for the Bell parameter through a concentration inequality such as the Azuma-Hoeffding inequality would be unable to pick out the small fraction of runs which do contain (close-to)-uniform inputs. 

In the protocol, we use Raz's independent-source extractor since it works when one of the sources has $O(\log N)$ entropy. Other constructions of extractors are known such as Li's extractor that works with two sources of polylogarithmic entropy, i.e., $O(\log^C N)$ for a large constant $C$. However, since the outputs of the device $\mathcal{A}^N\mathcal{B}^N$ have only randomness in $O(\log N)$ of the runs, the present structure of the protocol does not allow to use such extractors. This means that the second block $\mathtt{X}_2$ of the min-entropy source must have high min-entropy $k_2$ at least $\left(1/2 + \delta' \right)n$ to compensate for the low randomness certified from the device outputs. This is also why we use the seeded extractor with the smallest seed length $d = O(\log n)$ and correspondingly require the first block $\mathtt{X}_1$ of the min-entropy source to have sufficiently high min-entropy $k_1 = O(n^{\alpha})$ for $0< \alpha \leq 1$. 


\subsection{Bell tests with arbitrary min-entropy sources}
While the analyses in the previous sections have focused on extracting randomness from (two-)block min-entropy sources, they can also be used to show that Bell tests may be performed with single block sources of arbitrary min-entropy, overcoming a no-go result shown in \cite{TSS13}. Specifically in \cite{TSS13}, it was shown that given classical side information $\Lambda$ about a source, if the inputs $\mathcal{X}^N, \mathcal{Y}^N$ in a two-party $N$-round Bell test (with $m_A$ inputs for Alice and $m_B$ for Bob) inputs have min-entropy $H_{\min}\left(\mathcal{X}^N\mathcal{Y}^N | \Lambda \right) \leq N \log(m_A + m_B - 1)$ then no conclusion can be drawn from the Bell test since the no-signalling limit of the inequality can be saturated by local deterministic behaviours. From this (and the fact that at minimum $m_A = m_B = 2$ for a Bell test), it was deduced that any source of randomness with $H_{\min}(\mathcal{X}^N\mathcal{Y}^N | \Lambda) \leq N \log 3$ is useless as a source for Bell tests. We here show that the pre-processing of the source to produce a somewhere-random-source can be used to overcome this limitation, specifically any weak source of randomness of arbitrary min-entropy can be used for a Bell test. 

To do this, we will build the somewhere-random-source using the extractor in Thm. \ref{thm-first-step-1} from the given min-entropy source $\mathtt{X} \in \{0,1\}^n$ of arbitrary min-entropy $H^{\epsilon'}(\mathtt{X} | \Lambda) \geq k$. This extractor uses a seed of length $d = O(\log^3 n)$ and produces $m = k - 4 \log\left(\frac{1}{\epsilon}\right) - O(1)$ bits of randomness with an error $\epsilon + 2 \epsilon'$ for $\epsilon = \text{poly}\left(\frac{1}{n}\right)$. As in Section \ref{sec:SRS-Bell}, we choose an output length $m' = O(\log n)$ such that the output of the extractor has full support on this output, that is, it outputs each of the $m'$-bit strings with positive probability. For the Bell test, we do not care that $m'$ is of order $\log N$ where $N := 2^d \cdot m'/2$ since we do not aim to use the output bits as further inputs to an independent-source extractor. As shown in Sec. \ref{sec:MDL-Estimation}, any observed average violation of the MDL quantity $L_{m'}^j = \frac{2}{m'} \sum_{k_j=1}^{m'/2} \overline{M}^{j,k_j}_{\epsilon} \geq \delta$ guarantees that with high probability in at least a linear fraction of the runs within the $j$-th block, it holds that the true  MDL-Hardy parameter satisfies $\overline{M}^{j,k_j}_{\epsilon} \geq \delta/4 > 0$ which allows to detect non-locality. And this procedure works for any min-entropy source with min-entropy of order at least $\Omega(\log n)$.

\section{Acknowledgments}
Useful discussions with Pawe{\l} Horodecki, Stefano Pironio and Yuan Liu are acknowledged. R. R. acknowledges support from the Early Career Scheme (ECS) grant "Device-Independent Random Number Generation and Quantum Key Distribution with Weak Random Seeds" (Grant No. 27210620), the General Research Fund (GRF) grant "Semi-device-independent cryptographic applications of a single trusted quantum system" (Grant No. 17211122) and the Research Impact Fund (RIF) "Trustworthy quantum gadgets for secure online communication" (Grant No. R7035-21).


\section{Supplemental Material.}

\subsection{Randomness Extractors}
A (seeded) randomness extractor is a function $\text{Ext}$ that takes as input a weakly random source $X$ together with a uniformly distributed short seed $Y$ and outputs a string $Z$, such that $Z$ is almost uniformly distributed from the point of view of any adversary whenever the min-entropy of $X$ is greater than some threshold value $k$. The adversarial side information $E$ may be represented by the state of a classical or quantum system, the requirement of $Z$ being almost uniform conditioned on quantum $E$ is known to be strictly stronger than when $E$ is classical. The extractor is said to be strong if the output $Z$ is almost independent of the seed. See \cite{KR11, DPVR12} for a discussion on extractors.
\begin{definition}(Strong Extractor).
Let $n, d, m \in \mathbb{N}$, $0 \leq k \leq n$ and $\epsilon \in [0,1]$.
A function $\text{Ext:} \{0,1\}^n\times\{0,1\}^d\rightarrow\{0,1\}^m$ is a quantum proof $(k,\epsilon)-$strong extractor with uniform seed, if for all classical-quantum-states $\rho_{XE}$ with $H_{\min}(X|E)\geq k$ and uniform seeds $Y$ on $\{0,1\}^d$ independent of $\rho_{XE}$ we have 
\begin{equation}
    \frac{1}{2} \big\| \rho_{\text{Ext}(X,Y)YE}-\rho_{Um}\otimes\rho_Y\otimes\rho_E \big\| \leq\epsilon
\end{equation}
where $\rho_{Um}$ is the fully mixed state on a system of dimension $2^m$.
\end{definition}

Trevisan's extractor \cite{T01} is based on a construction that encodes the source $X$ using a list-decodable code $C$. The output of the extractor consists of certain bits of $C(X)$ specified by the seed and a construction called a weak design. 
\begin{definition}(Weak design).
A family of sets $R_1,\ldots,R_m\subset[d]$ is a weak $(t,r)$-design if
\begin{enumerate}
    \item  For all $i, |R_i|=t$. 
    \item  For all $i$,$\sum_{j=1}^{i-1} 2^{|R_i\cap R_j|}\leq rm$.
\end{enumerate}
\end{definition}
De et al. in \cite{DPVR12} prove that Trevisan's construction gives a quantum-proof strong extractor.
\begin{definition}(Trevisan's Extractor).
For a one-bit extractor $C:\{0,1\}^n\times\{0,1\}^t\rightarrow\{0,1\}$ which uses a seed of length $t$, and for a weak $(t,r)$-design $R_1,\ldots,R_m\subset[d]$ we define the m-bit extractor $\text{Ext}_c:\{0,1\}^n\times\{0,1\}^d\rightarrow\{0,1\}^m$
as
\begin{equation}
    \text{Ext}_c(x,y):=C(x,y_{R_1})\ldots C(x,y_{R_m})
\end{equation}
\end{definition}

\begin{theorem}\cite{DPVR12}.
Let $C:\{0,1\}^n\times\{0,1\}^t\rightarrow\{0,1\}$ be a $(k,\epsilon)$-strong extractor with uniform seed and $R_1,\ldots,R_m\subset[d]$ a weak $(t,r)$-design. Then Trevisan's Extractor
\begin{equation}
    \text{Ext}_c:\{0,1\}^n\times\{0,1\}^d\rightarrow\{0,1\}^m
\end{equation}
is a quantum-proof $(k+rm+\log\frac{1}{\epsilon},3m\sqrt{\epsilon})$-strong extractor.
\end{theorem}

They also gave different concrete constructions of quantum-proof extractors building on the above theorem, depending on whether one wants to maximise the output length, minimize the seed length etc. We highlight two specific constructions below, which can be used in the first step of the protocol depending on the min-entropy of the source and the requirement on the output and seed lengths.  
\begin{theorem}(Corollary 5.4 of \cite{DPVR12}).
\label{thm-first-step-1}
There exist a function $\text{Ext}_{\text{T}}:\{0,1\}^n\times\{0,1\}^d\rightarrow\{0,1\}^m$ which is a quantum-proof $(m+4\log\frac{1}{\epsilon}+O(1),\epsilon)$-strong extractor with uniform seed of length $d=O(\log^2(\frac{n}{\epsilon})\log m)$.
\end{theorem}
For $\epsilon = \text{poly}(\frac{1}{n})$ in this construction, the seed has length $d = O(\log^3 n)$. Given a source $X$ that has smooth min-entropy $H_{\min}^{\epsilon'}(X|E) \geq k$ conditioned on adversarial side information $E$, the extractor in Thm. \ref{thm-first-step-1} produces $m = k - 4 \log(\frac{1}{\epsilon}) - O(1)$ bits of randomness with an error $\epsilon+2\epsilon'$. 
\begin{theorem}(Corollary 5.6 of \cite{DPVR12}).
\label{thm-first-step-2}
If for constant $0 < \alpha \leq 1$, the source has min-entropy $H_{\min}(X|E) = n^{\alpha}$, then for any $0 < \gamma < \alpha$ there exists an explicit function $\text{Ext}_{\text{T}}: \{0,1\}^n \times \{0,1\}^d \rightarrow \{0,1\}^m$ which is a quantum-proof $(k, \epsilon)$-strong extractor with $k = n^{\gamma} m + 8 \log(m/\epsilon) + O(1)$, $d = O(\frac{1}{\gamma}\log n)$ and $\epsilon = n^{-\Omega(1)}$. 
\end{theorem}
Choosing $\gamma$ to be a constant gives the seed length to be $d = O(\log n)$. Given a source $X$ that has smooth min-entropy $H_{\min}^{\epsilon'}(X|E) \geq n^{\alpha}$ for $0 < \alpha \leq 1$, the extractor in Thm. \ref{thm-first-step-2} produces $m = n^{\alpha - \gamma} - o(1)$ bits of randomness with an error $\epsilon + 2 \epsilon'$.  

We also use the notion of \textit{independent-source extractors} that extract randomness from two independent weak sources rather than using a fully uniform seed. 
\begin{definition}
Let $n_1, n_2, m \in \mathbb{N}$, $0 \leq k_1 \leq n_1$, $0 \leq k_2 \leq n_2$ and $\epsilon \in [0,1]$. A function $\text{Ext:} \; \{0,1\}^{n_1} \times \{0,1\}^{n_2} \rightarrow \{0,1\}^m$ is said to be a two-source extractor strong in the $j$-th input (for $j \in \{1,2\}$), if for independent $\mathtt{X}_1, \mathtt{X}_2$ with $H_{\min}(\mathtt{X}_1) \geq k_1$ and $H_{\min}(\mathtt{X}_2) \geq k_2$ it holds that
\begin{equation}
 \frac{1}{2} \big\| \text{Ext}(\mathtt{X}_1,\mathtt{X}_2) - U_m \circ \mathtt{X}_j \big\| \leq\epsilon,
\end{equation}
where $U_m$ denotes the uniform random variable on $m$-bit strings. 
\end{definition}
As the final step of the protocol, we will use an independent-source extractor that is based on the construction by Raz \cite{Raz05} and proven to be quantum-proof in the Markov model in \cite{FPS16}. 
\begin{definition}
A ccq-state $\rho_{C_1C_2Q}$ belongs to the Markov model if $C_1 \leftrightarrow Q \leftrightarrow C_2$ forms a Markov chain, i.e., $I(C_1:C_2|Q)_{\rho} = 0$ where $I(C_1:C_2|Q)_{\rho} = H(C_1Q)_{\rho} + H(C_2Q)_{\rho} - H(C_1C_2Q)_{\rho} - H(Q)_{\rho}$ denotes the conditional mutual information with $H(Q)_{\rho} = - \text{Tr}\left[\rho_Q \log \rho_Q\right]$. 
\end{definition}
\begin{definition}
Let $n_1, n_2, m \in \mathbb{N}$, $0 \leq k_1 \leq n_1$, $0 \leq k_2 \leq n_2$ and $\epsilon \in [0,1]$. A function $\text{Ext:} \; \{0,1\}^{n_1} \times \{0,1\}^{n_2} \rightarrow \{0,1\}^m$ is said to be a quantum-proof two-source extractor strong in the $j$-th input (for $j \in \{1,2\}$), if for Markov source $\rho_{\mathtt{X}_1 \mathtt{X}_2Q}$ with $H_{\min}(\mathtt{X}_1 | Q) \geq k_1$ and $H_{\min}(\mathtt{X}_2|Q) \geq k_2$ it holds that
\begin{equation}
 \frac{1}{2} \big\| \rho_{\text{Ext}(\mathtt{X}_1,\mathtt{X}_2)X_jQ} - \rho_{U_m} \otimes \rho_{\mathtt{X}_jQ} \big\| \leq\epsilon,
\end{equation}
where $\rho_{U_m}$ denotes the fully mixed state of dimension $2^m$. 
\end{definition}
The following Theorem was proven in \cite{FPS16} and applied to the extractor constructed by Raz \cite{Raz05}.
\begin{theorem}(\cite{FPS16}).
Any $(k_1,k_2,\epsilon)$-strong two source extractor is a $(k_1+\log \frac{1}{\epsilon},k_2+\log \frac{1}{\epsilon},\sqrt{3\epsilon\cdot2^{M-2} })$-strong quantum proof extractor in the Markov model where $M$ is the output length of the extractor.
\end{theorem}
\begin{theorem}(\cite{FPS16}).
\label{thm-second-ext}
For any $n_1, n_2, k'_1, k'_2, m$ and any $0 < \delta' < 19/32$ such that
\begin{eqnarray}
n_1 &\geq & 6 \log n_1 + 2 \log n_2, \nonumber \\
k'_1 &\geq & \left(\frac{1}{2} + \delta' \right) n_1 + 3 \log n_1 + \log n_2, \nonumber \\
k'_2 &\geq & \frac{163}{32} \log \left(\left(1 + \frac{3\delta'}{19} \right) n_1 - k'_1 \right), \nonumber \\
m &\leq & \frac{16 \delta'}{19} \min \left[\frac{n_1}{8}, \frac{4k'_2}{163} \right] -1,
\end{eqnarray}
there exists an explicit function $\text{Ext}_{\text{R}}: \{0,1\}^{n_1} \times \{0,1\}^{n_2} \rightarrow \{0,1\}^m$ that is a quantum-proof $\left(k'_1, k'_2, \epsilon' \right)$-two-source extractor strong in both inputs (separately) with $\epsilon' = \frac{\sqrt{3}}{2} 2^{-m/4}$. 
\end{theorem}

\subsection{Entropy Accumulation}
\label{sec:Entropy-Acc}

The Entropy Accumulation Theorem (EAT) is an information-theoretic tool to bound the total conditional smooth min entropy $H_{\min}^{\gamma}\left( \mathcal{A}^N\mathcal{B}^N|\mathcal{X}^N\mathcal{Y}^N E \Lambda\right)$ for $\gamma\in(0,1)$ under the condition that the protocol did not abort.
In order to apply the theorem, one must first verify that for the Block Min-Entropy Source and the Device $D=(D_1,D_2)$, the sequential procedure of the protocol fulfils the requirement of the Entropy Accumulation Theorem. One must then devise a min-tradeoff function that quantifies the entropy accumulated in a single run of the protocol for given observed violation of the Bell inequality $\widetilde{G}$. After constructing this function, known techniques are used to derive a bound on $H_{\min}^{\gamma}\left( \mathcal{A}^N\mathcal{B}^N|\mathcal{X}^N\mathcal{Y}^N E \Lambda \right)$. For $N = 2^d \cdot (m'/2)$ rounds of the protocol in which $m'/2$ rounds are run with (close-to-)uniform inputs, the Entropy Accumulation statement gives $H_{\min}^{\gamma}\in\Omega(m')$ which is optimal. In the following, we recall the definitions of the EAT channels, Min-tradeoff function, and the EAT theorem from \cite{ADF+18, DOR16, KAF17}. 
\begin{definition}(EAT Channels).
The EAT Channels $\mathcal{EAT}_j: R_{j-1} : R_jA_jB_jX_jY_jD_j$ for $j \in [M]$ are completely positive trace-preserving (CPTP) maps such that 
\begin{enumerate}
\item $A_jB_jX_jY_jD_j$ are finite-dimensional classical random variables ($A_j B_j$ are measurement outcomes, $X_jY_j$ are measurement inputs, $D_j$ is a random variable evaluating the winning condition), $R_j$ are arbitrary quantum registers (holding information about the quantum state at the $j$-th round). 

\item For any input state $\rho_{R_{j-1}R'}$ where $R'$ is a register isomorphic to $R_{j-1}$, the classical value $D_j$ (the Bell indicator value at the $j$-th round) can be measured from the marginal $\rho_{A_jB_jX_jY_j}$ (the classical random variables at the round) of the output state $\rho_{R_jA_jB_jX_jY_jD_jR'} = \left(\mathcal{EAT}_j \otimes \mathcal{I}_{R'} \right)\left(\rho_{R_{j-1}R'} \right)$ without changing the state. 

\item For any initial state $\rho_{R_0E\Lambda}$, the final state $\rho_{\mathcal{A}^N\mathcal{B}^N\mathcal{X}^N\mathcal{Y}^N\mathcal{D}^NE\Lambda} = \left(\left( \text{Tr}_{R_N} \circ \mathcal{EAT}_{N} \circ \ldots \mathcal{EAT}_{1} \right) \otimes \mathcal{I}_E \right) \rho_{R_0E\Lambda}$ satisfies the Markov condition $\mathcal{A}^{l-1}\mathcal{B}^{l-1} \leftrightarrow \mathcal{X}^{l-1}\mathcal{Y}^{l-1}E\Lambda \leftrightarrow X_lY_l$ for each $l \in [N]$. Here $\mathcal{A}^{l-1} = A_1,A_2,\ldots, A_{l-1}$ and similarly for the other random variables. The condition states that the weak source does not change its state depending on previous outputs from the device so that future inputs do not reveal any new information about previous outputs. 
\end{enumerate}
\end{definition}
To use entropy accumulation, we need to verify that the protocol evolves the states using EAT channels. 
Denoting $l = (j,k_j)$, we have quantum registers $Q^{A}_{l}, Q^{B}_{l}$ holding the quantum states of the device for the $k_j$-th run within the $j$-th block, classical registers $X_{l}, Y_{l}, A_{l}, B_{l}$ for the inputs and outputs of the device and $D_{l}$ evaluating the outcome in the test. Let us denote by $\mathcal{C}_l$ the channels that evolve the states $\rho_{Q^{A}_{l-1}Q^{B}_{l-1}}$ to $\rho_{Q^{A}_{l}Q^{B}_{l}X_{l}Y_{l}A_{l}B_{l}D_{l}}$. We see that the channels $\mathcal{C}_l$ are indeed EAT channels.
\begin{enumerate}
\item The input-output $X_{l},Y_{l}, A_l, B_l$ are finite dimensional classical systems, and the $Q^{A}_{l}Q^{B}_{l}$ are quantum registers.

\item The $D_l$ is a classical function of the input-output registers $X_{l},Y_{l}, A_l, B_l$. 

\item The Markov chain condition $I\left(\mathcal{A}^{l-1} \mathcal{B}^{l-1}: X_l Y_l | \mathcal{X}^{l-1} \mathcal{Y}^{l-1} E \Lambda \right) = 0$ is satisfied by assumption. 

\end{enumerate}

Now, let us denote $\perp$ as the event of aborting the protocol and $\overline{\perp}$ as the complementary event of not aborting the protocol,
\begin{equation}
\overline{\perp} := \big\{ \wedge_{j=1}^{2^d} \left( L_{m'}^j \geq \delta \right) \big\}.
\end{equation}
Let $\rho = \rho^{\mathcal{A}^N\mathcal{B}^N\mathcal{X}^N\mathcal{Y}^N\mathcal{D}^NE\Lambda}$ denote the joint state of the devices held by the honest parties and Eve at the end of the $N := 2^d \cdot (m'/2)$ runs of the protocol. Let $\rho_{|\overline{\perp}}$ denote the state conditioned on the event of not aborting the protocol. 
Using the statement shown in \cite{KAF17} we have the following.
\begin{proposition}
\label{prop:rand-acc-2}
Let $\rho = \rho^{\mathcal{A}^N\mathcal{B}^N\mathcal{X}^N\mathcal{Y}^N\mathcal{D}^NE\Lambda}$ denote the joint state of the devices held by the honest parties and Eve at the end of the $N$ runs of the protocol, $\overline{\perp}$ denote the event of not aborting and let $\rho_{|\overline{\perp}}$ denote the state conditioned on the event of not aborting the protocol. Then for any $\epsilon_{EA}, \gamma \in (0,1)$, either the protocol aborts with probability greater than $1 - \epsilon_{EA}$ or it holds that
\begin{equation}
\label{eq:entropy-bound}
H_{\min}^{\gamma}\left(\mathcal{A}^N\mathcal{B}^N|\mathcal{X}^N\mathcal{Y}^NE\Lambda \right)_{\rho_{|\overline{\perp}}} \geq (m'/2) \cdot g\left(\gamma, \epsilon_{EA}, \delta, m'/2, \epsilon \right) - O(\log(1/\gamma)),
\end{equation}
where
\begin{equation}
g\left(\gamma, \epsilon_{EA}, \delta, m'/2, \epsilon \right) := \max_{0 < s_t < (1/4 - \epsilon^2)^2 \frac{\sqrt{2}-1}{2}} \left[f_{\min}\left( \delta, s_t \right) - \sqrt{\frac{2}{m'}} 2 \left( \log 9 + a(s_t) (1/2+\epsilon)^2 \right) \sqrt{1 - 2 \log{(\gamma \cdot \epsilon_{EA})}} \right],
\end{equation}
with 
\[ f_{\min}(p, s_t) =\begin{cases} 
      \alpha_{\epsilon}(p) & S_{\epsilon}(p) \leq s_t \\
      a(s_t) S_{\epsilon}(p) + b(s_t) & S_{\epsilon}(p) > s_t
   \end{cases}
\] 
and
\begin{equation}
a(s_t) = \frac{d}{d S_{\epsilon}(p)} \alpha_{\epsilon}(p) \bigg\vert_{S_{\epsilon}(p) = s_t}, \qquad b(s_t) = \alpha_{\epsilon}(s_t) - a(s_t) \cdot s_t,
\end{equation}
and
\begin{eqnarray}
S_{\epsilon}(p) &=& (1/2 - \epsilon)^2 p(1) - (1/2 + \epsilon)^2 p(-1), 
\end{eqnarray}
\[ g_{\epsilon}(p) =\begin{cases} 
      1 - h\left[\frac{1}{2} + \frac{1}{(1/4- \epsilon^2)^2} \sqrt{S_{\epsilon}(p)\left( S_{\epsilon}(p) + (1/4- \epsilon^2)^2 \right)} \right] & \frac{S_{\epsilon}(p)}{(1/4- \epsilon^2)^2} \in \left[0, \frac{\sqrt{2}-1}{2} \right) \\
      1 & \frac{S_{\epsilon}(p)}{(1/4- \epsilon^2)^2} \in \left[ \frac{\sqrt{2}-1}{2}, 1 \right]
   \end{cases}
\] 
\end{proposition}
The proof follows from the statement shown in \cite{KAF17} where it was shown that for a single block $j^*$ with inputs $\epsilon$-close-to-uniform, the lower bound in the Lemma holds for $H_{\min}^{\gamma'}\left(\mathcal{A}^{l^*}\mathcal{B}^{l^*}|\mathcal{X}^{l^*}\mathcal{Y}^{l^*}E\Lambda\right)$ where $l^* = (m'/2)\cdot j^*$ denotes the final run within the $j^*$-th block, and we choose $\gamma' = \gamma/4$ for $\gamma \in (0,1)$. Specifically, we use the statement from \cite{KAF17} treating the inputs from the previous blocks $\mathcal{X}^{l^*-1}\mathcal{Y}^{l^*-1}$ as part of the adversarial classical side information. Using the Markov assumption \ref{assum:Markov} that $I\left(\mathcal{A}^{l^*}\mathcal{B}^{l^*}; X_{l^*+1}Y_{l^*+1}|\mathcal{X}^{l^*}\mathcal{Y}^{l^*}E\Lambda\right) = 0$ meaning that $I\left(\mathcal{A}^{l^*}\mathcal{B}^{l^*}; \overline{X}^{l^*}\overline{Y}^{l^*}|\mathcal{X}^{l^*}\mathcal{Y}^{l^*}E\Lambda\right) = 0$ where $\overline{X}^{l^*} = \mathcal{X}^N\setminus X^{l^*} = X_{l^*+1} \ldots X_N$, the lower bound holds also for $H_{\min}^{\gamma'}\left(\mathcal{A}^{l^*}\mathcal{B}^{l^*}|\mathcal{X}^{N}\mathcal{Y}^NE\Lambda\right)$. Finally, the bound can be extended to $H_{\min}^{\gamma}\left(\mathcal{A}^N\mathcal{B}^N|\mathcal{X}^N\mathcal{Y}^NE\Lambda \right)$ using the chain rule for smooth min-entropy \cite{DBWR14, VDTR13},
\begin{equation}
H_{\min}^{\gamma}\left(\mathcal{A}^N\mathcal{B}^N|\mathcal{X}^N\mathcal{Y}^NE\Lambda \right)_{\rho_{|\overline{\perp}}} \geq H_{\min}^{\gamma/4}\left(\mathcal{A}^{l^*}\mathcal{B}^{l^*}|\mathcal{X}^N\mathcal{Y}^NE\Lambda \right)_{\rho_{|\overline{\perp}}} + H_{\min}^{\gamma/4}\left(\overline{\mathcal{A}^{l^*}} \overline{\mathcal{B}^{l^*}}|\mathcal{X}^N\mathcal{Y}^NE\Lambda \right)_{\rho_{|\overline{\perp}}} - O\left(\log(1/\gamma) \right),
\end{equation}
where $\overline{\mathcal{A}^{l^*}} = \mathcal{A}^N \setminus \mathcal{A}^{l^*}$ denotes the outputs in the blocks after $j^*$ and the last term is specifically $\log\left(\frac{1}{1 - \sqrt{1 - \gamma^2/16}}\right)$. We now bound the second term by $0$ since the inputs to these blocks may not be random, to obtain Eq. \ref{eq:entropy-bound}.

\subsection{An alternative Bell test with improved parameters}
After the first step of the protocol, we have $(2^d,m,\epsilon)$-quantum-somewhere-random source. The second step in the protocol is to use this source as the input for $\big\lfloor 2^d \cdot m/m' \big \rfloor$ rounds of a suitably designed Bell test with $m'$-bit inputs per round. 

Here, we will illustrate the protocol using the $3$-player GHZ game which achieves improved parameters over the $2$-player MDL-Hardy game. 
In the $3$-player GHZ game, each player receive a single bit input ($x,y,z\in\{1,2\}$) and produces a single bit output ($a,b,c\in\{0,1\}$). 

For the correlators defined as
\begin{equation}
    \langle A_X B_Y C_Z\rangle := \sum_{\substack{a,b,c \\ a \oplus b \oplus c = 0}} P_{A,B,C|X,Y,Z}(a,b,c|x,y,z)-\sum_{\substack{a,b,c \\ a \oplus b \oplus c = 1}} P_{A,B,C|X,Y,Z}(a,b,c|x,y,z),
\end{equation}
in the GHZ-Mermin game, the goal of the players is to maximize the Mermin expression $M=\langle A_1B_1C_1\rangle-\langle A_1B_2C_2\rangle-\langle A_2B_1C_2\rangle-\langle A_2B_2C_1\rangle$. It is well known that $M\leq2$ holds for classical theories (Local Hidden Variable models) while Quantum theory achieves the maximum (algebraic) value of $M=4$. The optimal quantum strategy (to achieve $M=4$) is for the players to measure $A_1=B_1=C_1=\sigma_x$ and $A_2 = B_2 = C_2 = \sigma_y$ on the GHZ state $|\phi_{\text{GHZ}} \rangle=\frac{1}{\sqrt{2}}(|000\rangle+|111\rangle)$.

In \cite{WBA18}, Woodhead et al. derived the guessing probability for the 3-player GHZ game as 
\begin{equation}
    P_{g}(A_1|E)\leq f(M)
\end{equation}
where $P_{g}(A_1|E)$ denotes any quantum adversary Eve's probability of guessing Alice's measurement outcome for measurement $A_1$, with $f$ being the function
\begin{equation}
f(M)=
\left\{\begin{array}{cc}
\frac{1}{2}+\frac{1}{2} \sqrt{M(1-\frac{M}{4})} \quad &\text{for}\quad M\geq2+\sqrt{2} \\
1+\frac{1}{\sqrt{2}}-\frac{M}{4} \quad &\text{for} \quad  M\leq 2+\sqrt{2}
\end{array}
\right.
\end{equation}
for the range $2\sqrt{2}\leq M \leq 4$. Similar bounds hold for the local guessing probabilities $P_{g}(B_y|E),P_{g}(C_z|E)$ and $P_{g}(A_x|E)$. Note that $P_{g}(A_1|E)$ is given by 
\begin{equation}
    P_{g}(A_1|E)=\max_{\{M_E^a\}_a} \bigg|\sum_a p_a \text{Tr}(M_E^a\rho_E^a) \bigg|
\end{equation}
where the maximisation is over all POVMs ${\{M_E^a\}_a}$ on the eavesdropper's quantum system $E$. The conditional min-entropy $H_{\min}(A_x|E)$ is related to the guessing probability as 
\begin{equation}
    H_{\min}(A_1|E)=-\log_2 P_{g}(A_1|E)
\end{equation}
The conditional min-entropy is a lower bound for the conditional von Neumann entropy $H(A_1|E)$. 



When we use the $(2^d, m, \epsilon)$-quantum-somewhere-random source to choose inputs to the GHZ game, the individual bits $x, y, z$ are $\epsilon$-away from uniform for each run of the GHZ Bell test. We therefore test the Measurement-Dependent-Locality version of the GHZ game, where we evaluate the parameter $\bar{M}_{\epsilon}$ defined as
\begin{equation}
\bar{M}_{\epsilon} := \sum_{a,b,c,x,y,z} \nu(x,y,z) M_{\epsilon}(a,b,c,x,y,z) P_{A,B,C|X,Y,Z}(a,b,c|x,y,z),
\end{equation}
where $M_{\epsilon}$ is an indicator taking values given by
\[ M_{\epsilon}(a,b,c,x,y,z) := \begin{cases} 
      \left(1/2-\epsilon \right)^3 & \left[(x,y,z) = (1,1,1) \right] \wedge \left[a \oplus b \oplus c = 0 \right] \\
    -\left(1/2 + \epsilon \right)^3 & \left[(x,y,z) \in \{ (1,2,2), (2,1,2), (2,2,1) \} \right] \wedge \left[a \oplus b \oplus c = 0 \right] \\
      0 & \text{otherwise}
   \end{cases}
\]
and $\nu(x,y,z)$ is the probability distribution of the bits $x,y,z$ being chosen as the inputs from the somewhere-random source. We have that
\begin{equation}
\left(1/2-\epsilon \right)^3 \leq \nu(x,y,z) \leq \left(1/2 + \epsilon \right)^3 \quad \forall (x,y,z).
\end{equation}
The maximum value of $\bar{M}_{\epsilon}$ over classical theories (local hidden variable strategies) is seen to be $0$, while the corresponding maximum value in quantum theory is $\left(1/2 - \epsilon \right)^3 \nu(1,1,1) \geq \left(1/2 - \epsilon \right)^6$. Note that
\begin{eqnarray}
\bar{M}_{\epsilon} &\leq & \left(1/2 - \epsilon \right)^3 \left(1/2 + \epsilon \right)^3 \bigg[ \sum_{\substack{a,b,c \\ a \oplus b \oplus c = 0}} \bigg(P_{A,B,C|X,Y,Z}(a,b,c|1,1,1) - P_{A,B,C|X,Y,Z}(a,b,c|1,2,2) - \nonumber \\ && \qquad \qquad \qquad \qquad \qquad \qquad \qquad P_{A,B,C|X,Y,Z}(a,b,c|2,1,2)- P_{A,B,C|X,Y,Z}(a,b,c|2,2,1) \bigg) \bigg] \nonumber \\
&\leq &  \left(1/4 - \epsilon^2 \right)^3 (M-2)/2.
\end{eqnarray}
So that $M \geq \frac{2 \bar{M}_{\epsilon}}{\left(1/4 - \epsilon^2 \right)^3} + 2$. 
In a real experiment with noise, suppose that the value achieved is $\bar{M}_{\epsilon} \geq \delta$ for some constant $\delta > 0$. We then obtain that 
\begin{equation}
\label{eq:GHZ-guess}
P_{g}(A_1 | E) \leq f\left(\frac{2 \delta}{\left(1/4 - \epsilon^2 \right)^3} + 2 \right),
\end{equation}
where we have used the fact that $f(M)$ is monotonically non-increasing for $2 \leq M \leq 4$ (it is monotonically decreasing for $2 + \sqrt{2} \leq M \leq 4$). Using the EAT we can similarly obtain a lower bound on $H_{\min}^{\gamma}\left(\mathcal{A}^N|\mathcal{X}^N\mathcal{Y}^NE\Lambda\right)_{\rho_{|\overline{\perp}}}$ as $(m'/3) f_{\min}(\delta) - c \sqrt{m'/3}$ where $c   = 2 \left(\log 5 + \big\lceil \| \nabla f_{\min} \|_{\infty} \big \rceil \right) \sqrt{1 - 2 \log \left(\gamma \cdot \epsilon_{EA}\right)}$ where $f_{\min}(\delta)$ is approximately $-\log_2 P_g(A_1|E)$ from Eq.(\ref{eq:GHZ-guess}) (for $\delta$ bigger than a suitable threshold). Other candidate Bell inequalities, in particular the ones from \cite{SSKA21, WBC22, ZRLH22} need to be investigated in future for better yield, robustness and feasibility. Further improvements in parameters can be achieved following the techniques for reducing the extractor error in \cite{RRV99} and will be pursued in future work.

\subsection{Estimation of the Bell-MDL parameter}
\label{sec:MDL-Estimation}
In the protocol, we use the $(2^d, m', \epsilon)$-quantum-somewhere-random source to choose the inputs for $\lfloor 2^d \cdot m'/2 \rfloor$ rounds of the MDL-Hardy test. In other words, supposing that $m'$ is even, we use the bits $S_j = \text{Ext}(X,j)$ for $j = 1, \ldots, 2^d$ to choose the inputs in $m'/2$ sequential runs of the test. At the end of the procedure, the parties estimate the MDL-Hardy parameter $\bar{M}_{\epsilon}^{j,k_j}$ for each of the $k_j \in [m'/2]$ runs within each of the $j \in [2^d]$ sequential blocks. They abort the protocol unless $L_{m'}^j = \frac{2}{m'} \sum_{k_j=1}^{m'/2} \bar{M}^{j,k_j}_{\epsilon} \geq \delta$ for each of the $2^d$ blocks, for some constant $\delta > 0$. We use the following Lemma based on the Azuma-Hoeffding inequality to estimate the average value of the parameter $\bar{M}_{\epsilon}$ for the conditional boxes over all runs of the protocol. 
\begin{lemma}
Consider arbitrary random variables $W_i$ for $i = 0,1,...,m'/2$, and binary random variables $B_i$ for $i = 1,...,m'/2$ that are functions of $W_i$, i.e. $B_i = f_i(W_i)$ for some functions $f_i$. Denote the conditional means as $\overline{B}_i = \mathbb{E}(B_i | W_{i-1}, \ldots, W_1, W_0)$ for $i = 1,\ldots, m'/2$. Define for $k = 1, \ldots, m'/2$, the empirical average 
\begin{equation}
L_k = \frac{1}{k} \sum_{i=1}^k B_i,
\end{equation}
and the empirical average of conditional means
\begin{equation}
\overline{L}_k = \frac{1}{k} \sum_{i=1}^k \overline{B}_i.
\end{equation}
Then it holds that
\begin{equation}
\text{Pr}\left[|L_{m'} - \overline{L}_{m'}| \geq \delta_{Az} \right] \leq 2 \exp{\left(-m' \delta_{Az}^2/4\right)}.
\end{equation}
\end{lemma}
The above Lemma states that with high probability (specifically with probability $1-2 \exp^{-m' \delta^2/16}$), the arithmetic average of the conditional boxes in the $j$-th block satisfies $\overline{L}^{j}_{m'} = \frac{2}{m} \sum_{k_j=1}^{m/2} \overline{M}^{j,k_j}_{\epsilon} \geq \delta/2$ when the observed value satisfies $L^j_{m'} = \frac{2}{m} \sum_{k_j=1}^{m/2} M^{j,k_j}_{\epsilon}  \geq \delta$. We also state the following useful fact.
\begin{lemma}
If the arithmetic average satisfies $\overline{L}_{m'}^j =  \frac{2}{m'} \sum_{k_j=1}^{m'/2} \overline{M}^{j,k_j}_{\epsilon} \geq \delta/2$ for constant $\delta > 0$, with $\overline{M}^{j,k_j}_{\epsilon} \leq 1/2$ for every $k = 1, \ldots, m'/2$, then in at least $ \frac{m'(\delta - 2\kappa)}{2(1 - 2 \kappa)}$ positions $k$, we have that $\overline{M}^{j,k_j}_{\epsilon} \geq \kappa$ for $0 < \kappa < \delta/2$. 
\end{lemma}
\begin{proof}
Let $\overline{L}_{m'}^j =  \frac{2}{m'} \sum_{k_j=1}^{m'/2} \overline{M}^{j,k_j}_{\epsilon} \geq \delta/2$ for constant $\delta > 0$, with $\overline{M}^{j,k_j}_{\epsilon} \leq 1/2$ for every $k = 1, \ldots, m'/2$. Consider the set $I := \big\{ k_j | \overline{M}^{j,k_j}_{\epsilon} \geq \kappa \big\}$. Then we have that
\begin{eqnarray}
\sum_{k_j \in I} \overline{M}^{j,k_j}_{\epsilon}+ \sum_{k_j \notin I} \overline{M}^{j,k_j}_{\epsilon} &\geq & \frac{m' \delta}{4} \nonumber \\
\implies \frac{1}{2} |I| + \kappa (m'/2 - |I|) &\geq & \frac{m' \delta}{4} \nonumber \\
\implies |I|  \geq  \frac{m'(\delta - 2\kappa)}{2(1 - 2 \kappa)}.
\end{eqnarray}
The two Lemmas above together show that when the observed value satisfies $L_m^j = \frac{2}{m'} \sum_{k_j=1}^{m;/2} M^{j,k_j}_{\epsilon}  \geq \delta$, then in at least $ \frac{m'(\delta - 2\kappa)}{2(1 - 2 \kappa)}$ out of the $m'/2$ runs within the $j$-th block, the true MDL-Hardy parameter satisfies $\overline{M}^{j,k_j}_{\epsilon} \geq \kappa$. In particular, choosing $\kappa = \delta/4$ we obtain that in at least $\frac{m' \delta}{2(2 - \delta)}$ runs it holds that $\overline{M}^{j,k_j}_{\epsilon} \geq \delta/4$ with high probability.

\end{proof}

\end{document}